\documentclass[]{article}
\usepackage{cite}
\usepackage{amsmath,amssymb,amsfonts,dsfont}
\usepackage{algorithmic}
\usepackage{graphicx}
\usepackage{textcomp}
\usepackage{caption}
\usepackage{subcaption}
\usepackage{bbm}
\usepackage{xcolor}
\usepackage{tikz}
\usepackage{cancel}
\usepackage[T1]{fontenc}
\usepackage{comment}
\usepackage{amsthm}

\newcommand{\margin}[1]{\marginpar{\tiny\color{blue} #1}}

\begin{document}
	\title{A Coupled Friedkin-Johnsen Model \\ of Popularity Dynamics in Social Media}
	\author{Gaya Cocca\thanks{G.~Cocca was a master's student of Politecnico di Torino, Corso Duca Degli Abruzzi 24, 10129, Torino, Italy. 
			E-mail: s316763@studenti.polito.it.}, Paolo Frasca\thanks{P.~Frasca is with Univ.\ Grenoble Alpes, CNRS, Inria,  Grenoble INP, GIPSA-lab, F-38000 Grenoble, France. E-mail: paolo.frasca@gipsa-lab.fr.}, Chiara Ravazzi\thanks{C.~Ravazzi is with the National Research Council of Italy (CNR-IEIIT), c/o Politecnico di Torino, Corso Duca Degli Abruzzi 24, 10129, Torino, Italy. 
			E-mail: chiara.ravazzi@cnr.it.}
		\thanks{This work has been supported by ANR grant FeedingBias (ANR-22-CE38-0017-01) and %has been funded 
        by the European Union – Next Generation EU, Mission 4, Component 1, under the PRIN project TECHIE: ``A control and network-based approach for fostering the adoption of new technologies in the ecological transition'' Cod. 2022KPHA24 CUP Master: D53D23001320006, CUP: B53D23002760006.}
	}
	\newtheorem{definition}{Definition}
	\newtheorem{example}{Example}
	\newtheorem{ass}{Assumption}
	\newtheorem{trule}{Rule}
	\newtheorem{corollary}{Corollary}
	\newtheorem{theorem}{Theorem}
	\newtheorem{lemma}{Lemma}
	\newtheorem{assumption}{Assumption}
	\newtheorem{proposition}{Proposition}
	\newtheorem{remark}{Remark}
	\newtheorem{problem}{Problem}
	\newtheorem{conjecture}{Conjecture}
	%%%%%%%%%%%%%%%%%%
	
	\newcommand{\Ecal}{\mathcal{E}}
	\newcommand{\Gcal}{\mathcal{G}}
	\newcommand{\Mcal}{\mathcal{M}}
	\newcommand{\Ncal}{\mathcal{N}}
	\newcommand{\Vcal}{\mathcal{V}} 
	\newcommand{\mc}{\mathcal}
	\newcommand{\be}{\begin{equation}}
		\newcommand{\ee}{\end{equation}}
	\newcommand{\ba}{\begin{array}}
		\newcommand{\ea}{\end{array}}
	\newcommand{\tcr}{\textcolor{red}}
	
	\newcommand{\norm}[1]{\|#1\|}
	\newcommand{\abs}[1]{|#1|}
	\newcommand{\rk}{\mathrm{rank}}
	\newcommand{\complex}{\mathbb{C}}
	\newcommand{\real}{\mathbb{R}}
	\newcommand{\realpositive}{\mathbb{R}_{>0}}
	\newcommand{\realnonnegative}{\mathbb{R}_{\geq0}}
	\renewcommand{\natural}{\mathbb{N}}
	\newcommand{\integer}{\mathbb{Z}}
	\newcommand{\integernonnegative}{\mathbb{Z}_{\ge 0}}
	\newcommand{\setdef}[2]{\{#1 \, : \; #2\}}
	\newcommand{\map}[3]{#1: #2 \rightarrow #3}
	\renewcommand{\l}{\left}
	\renewcommand{\r}{\right}
	\newcommand{\R}{\mathbb{R}} % used for real
	\newcommand{\N}{\mathbb{N}}  % used for natural
	\newcommand{\Z}{\mathbb{Z}}  % used for integer
	\newcommand{\G}{\mathcal{G}} % graph
	\newcommand{\V}{\mathcal{V}} % node set
	\newcommand{\E}{\mathcal{E}} % edge set
	
	\newcommand{\neigh}{ \mathcal{N}} 	% neighborhood
	\newcommand{\card}[1]{|#1|}  	% cardinality
	
	\newcommand{\Exp}{\mathds{E}} % expectation
	\def\Exp{\mathbb{E}}
	\def\Prob{\mathbb{P}}
	
	\newcommand{\ave}{\textup{ave}}\newcommand{\loc}{{\subscr{x}{loc}^{\star}}} % localization vector
	\newcommand{\op}{{\subscr{x}{opd}^{\star}}} % opinion profile
	
	\newcommand{\VAR}{\mathrm{VAR}} % variance
	\newcommand{\X}{\mathcal{X}} % generic set
	\newcommand{\e}{\textbf{e}} % vector of canonical basis
	\newcommand{\1}{\mathds{1}} % vector of ones
	\newcommand{\ind}{\mathds{1}} % indicator function
	\newcommand{\diag}{\operatorname{diag}}  
	\newcommand{\by}{\boldsymbol{y}}  
	\newcommand{\vy}{{\boldsymbol{y}}}  
	\newcommand{\bx}{\boldsymbol{x}}  
	\newcommand{\bc}{\boldsymbol{c}}  
	\newcommand{\bpi}{\boldsymbol{\pi}}  
	\newcommand{\bW}{\boldsymbol{W}}  
	\newcommand{\bP}{\boldsymbol{P}}  
	\newcommand{\bU}{\boldsymbol{U}}  
	\newcommand{\bTheta}{\boldsymbol{\Theta}}  
	\newcommand{\bId}{\boldsymbol{I}}
	\newcommand{\btheta}{\boldsymbol{\theta}}  
	\newcommand{\bM}{\boldsymbol{M}}  
	\newcommand{\veta}{\boldsymbol{\eta}}  
	
	\newcommand{\argmax}[1]{\underset{#1}{\mathrm{argmax\,}}}
	\newcommand{\argmin}[1]{\underset{#1}{\mathrm{argmin\,}}}

	\newcommand{\eps}{\varepsilon} % for small quantities

	\newcommand{\supp}{\mathrm{supp}}
	\newcommand{\interior}[1]{\overset{\small{\circ}}{#1}}
	\newcommand{\Sup}[1]{\underset{#1}{\mathrm{sup\,}}}
	\newcommand{\Max}[1]{\underset{#1}{\mathrm{max\,}}}
	\newcommand{\Lim}[1]{\underset{#1}{\mathrm{lim\,}}}
	\newcommand{\vx}{\boldsymbol{x}}
	\newcommand{\vQ}{\boldsymbol{Q}}
	\newcommand{\ve}{\boldsymbol{e}}
	\newcommand{\vr}{\boldsymbol{r}}
	\newcommand{\vb}{\boldsymbol{b}}
	\newcommand{\vf}{\boldsymbol{f}}
	\newcommand{\va}{\boldsymbol{a}}
	\newcommand{\vxi}{\boldsymbol{\xi}}
	\newcommand{\vl}{\boldsymbol{l}}
	\newcommand{\vz}{\boldsymbol{z}}
	\newcommand{\vk}{\boldsymbol{k}}
	\newcommand{\vu}{\boldsymbol{u}}
	\newcommand{\vv}{\boldsymbol{v}}
	\newcommand{\vw}{\boldsymbol{w}}
	\newcommand{\vc}{\boldsymbol{c}}
	\newcommand{\vp}{\boldsymbol{p}}
	\newcommand{\vq}{\boldsymbol{q}}
	\newcommand{\vd}{\boldsymbol{d}}
	\newcommand{\vh}{\boldsymbol{h}}
	\newcommand{\vm}{\boldsymbol{m}}
	\newcommand{\vg}{\boldsymbol{g}}
	\newcommand{\vX}{\boldsymbol{X}}
	\newcommand{\vsigma}{\boldsymbol{\sigma}}
	\newcommand{\vSigma}{\boldsymbol{\Sigma}}
	\newcommand{\vlambda}{\boldsymbol{\lambda}}
	\newcommand{\vomega}{\boldsymbol{\omega}}
	\newcommand{\vtheta}{\boldsymbol{\theta}}
	\newcommand{\vP}{\boldsymbol{P}}
	\newcommand{\vA}{\boldsymbol{A}}
	\newcommand{\valpha}{\boldsymbol{\alpha}}
	\newcommand{\vgamma}{\boldsymbol{\gamma}}
	\newcommand{\vG}{\boldsymbol{G}}
	\newcommand{\vN}{\boldsymbol{N}}
	\newcommand{\vV}{\boldsymbol{V}}
	\newcommand{\vU}{\boldsymbol{U}}
	\newcommand{\vM}{\boldsymbol{M}}
	\newcommand{\vH}{\boldsymbol{H}}
	\newcommand{\vOmega}{\boldsymbol{\Omega}}
	\newcommand{\vZ}{\boldsymbol{Z}}
	\newcommand{\vzeta}{\boldsymbol{\zeta}}
	\newcommand{\vY}{\boldsymbol{Y}}
	\newcommand{\vPhi}{\boldsymbol{\Phi}}
	\newcommand{\veceta}{\boldsymbol{\eta}}

	\newcommand{\figref}[1]{Fig.~\ref{#1}}
	\newcommand{\secref}[1]{Sec.~(\ref{#1}}
	\newcommand{\vect}[1]{\mathbf{#1}}
	\newcommand{\dsum}{\displaystyle\sum}
	\newcommand{\ie}[0]{{\em i.e., }}
	\newcommand\imag{\operatorname{i}}
	\newcommand\dd{\operatorname{d}}

	\newcommand{\field}[1]{\mathds{#1}}
	\newcommand{\Msf}{\Lambda}
	
	\newcommand{\Commento}[1]{\\---\textit{\textbf{#1}}---\\}
	\pagestyle{empty} % Removes all the page numbers (except for the title page)
	\maketitle
	\thispagestyle{empty} % Removes the page number in the first page
	
	\begin{abstract}
Popularity dynamics in social media depend on a complex interplay of social influence between users and popularity-based recommendations that are provided by the platforms.
In this work, we introduce a discrete-time dynamical system to model the evolution of popularity on social media. Our model generalizes the well-known Friedkin-Johnsen model to a set of influencers vying for popularity.
We study the asymptotic behavior of this model and illustrate it with numerical examples. Our results highlight the interplay of social influence, past popularity, and content quality in determining the popularity of influencers.
% 	\margin{questo mi sembra più adatto a introduzione, ma non capisco bene. hai una citazione?}
% {\color{magenta}{Motivated by the observation of popularity dynamics on digital platforms, in this paper we propose a model to explain the dichotomy between consensus and disagreement of appreciation towards influencers in competition. On platforms such as TikTok or Instagram, users tend to appreciate creators' content uniformly, regardless of contents' quality. In contrast, on YouTube, Reddit, or Spotify, the perceived quality of content plays a key role. Users' preferences tend to fragment, leading to the formation of communities. The proposed model aims to capture this transition, showing how quality is the key factor determining the balance between uniformity and diversification of perceived popularity.}}	
\end{abstract}
	
%\begin{IEEEkeywords}
%Network analysis and control; Social networks; Opinion dynamics
%\end{IEEEkeywords}
%

% Nowadays, social networks and digital platforms, offering countless types of content at our fingertips and reaching a huge number of users instantly, are increasingly influencing our daily lives. 
% In many cases, due to social contagion and recommendation systems that suggest already popular content based on previous success, it can happen that some content goes viral and a small number of influencers dominate the landscape, capturing most of the public's attention, while the rest fight for visibility. However, given the growing volume of content and the number of competing influencers, the dynamics of popularity on social platforms are increasingly complex to understand.

\section{Introduction}
Social media play a prominent role in contemporary societies. Therefore, it is important to understand how the popularity of contents, topics, and influencers evolves therein. 
Over the years, various kinds of models from different scientific fields have been proposed to describe the dynamics of popularity. These models include epidemic models\cite{Richier}, self-exciting processes\cite{Crane_2008}, and Bass-like diffusion models~\cite{gao}. Although building upon different basic principles, most of them emphasize that popularity is shaped by social influence between users and by the recommendation systems deployed by the platforms~\cite{bressan2016limits}.
Significant scholarly work, including mathematical models, has been produced on these topics.
Within the large literature about social influence, which in part predates the rise of social media, we single out the celebrated Friedkin-Johnsen model~\cite{FriedJon} and its variations, including those that account for multiple coupled dynamics~\cite{friedkin2016network}. %,parsegov2017}.
Regarding the recommendation systems, even though their specifics are not disclosed by the platforms, they are known to leverage past popularity and user preferences with the goal of increasing user engagement~\cite{Covington16,castaldo2021junk,castaldo2023online}.
The objective of the paper is to propose and study a conceptual model, based on ``first principles'', which has the goal of exploring how the interplay of social influence and recommendations can play out during the processes of content consumption in social media. 
Our contribution lies in the formulation and analysis of an original model for the evolution of popularity on social media. Our model describes the popularity of influencers as the aggregate of the attention that their contents receive from a population of users. The attention of each user towards each influencer evolves through a nonlinear dynamics that takes into account three key ingredients: 
(i) the social influence due to interactions between users; (ii) the competition among influencers, mediated by the popularity-based recommendations by the platform; and
(iii) the intrinsic quality of the contents produced by each influencer. The model is akin to a set of coupled Friedkin-Johnsen dynamics: each of them describes the attention that the users give to one influencer, while the nonlinear coupling accounts for the aggregate popularity of the influencers.

Our model is able to capture two distinct scenarios. When content quality is irrelevant, the attention levels that are given to influencers converge to consensus, that is, tend to become uniform across users: popularity is only driven by social contagion effects and by the recommendations that amplify the popularity of already popular content. Instead, when quality is relevant, user attention converges toward a more heterogeneous limit profile, and quality determines popularity. 
This dichotomy highlights how sensitivity to quality can serve as a key factor in shaping the collective dynamics of popularity on social media.

This paper is organized as follows.
Section~\ref{section:model} describes the mathematical model and  Section~\ref{section:results} studies its steady-state behavior. We first concentrate on the degenerate cases when there is no network effect or no quality effect, then consider the generic case with all effects. Section~\ref{sect:discussion} discusses the results, their relations and limitations, and their potential implications for the study of popularity of social media. Section~\ref{sect:conclusion} provides some conclusions and directions for future work.

\section{The model} \label{section:model}
We consider a set $\mathcal{V}$ of $n$ users of a social media platform: each of them can access contents from a set $\mathcal{I} $ of influencers at each instant of time $t\in\mathbb{N}\cup\{0\}$. 
Each user $v\in\mathcal{V}$ is endowed with a scalar value $x_v^{(i)}(t)\in[0,1]$, which quantifies their attention to contents posted by influencer $i\in\mathcal{I}$ at time $t$. Our objective is to develop a model to describe the evolution of this attention over time, capturing the interplay of the following factors:
\begin{itemize}
\item the social influence between users, interconnected by a social network;
\item the effect of recommendations proposed by the platform according to the influencers' previous popularity; 
\item the intrinsic quality of the contents produced by the influencers.
\end{itemize}

\textit{Modeling network effects:} By network effects, we mean the phenomenon whereby a user's attention  towards an influencer is directly or indirectly influenced by the attention given by other users. 
The ability of users to interact and influence each other can be formalized using a graph 
$\mathcal{G}=(\mathcal{V},\mathcal{E})$, where $\mathcal{E}\subseteq\mathcal{V}\times\mathcal{V}$ describes the potential interactions. 
We will use a \textit{weighted-adjacency matrix $P$}, associated with $\mathcal{G}$, in order to encode the intensity of these interactions. The matrix $P$ is \textit{row stochastic}, i.e., has nonnegative entries and $P\mathds{1}_n=\mathds{1}_n$ (where $\mathds{1}_n$ denotes a vector with $n$ entries, all equal to $1$). A node \( v\in\mathcal{V} \) in a directed graph is said to be \emph{aperiodic} if the greatest common divisor of the lengths of all cycles starting and ending at \( v \) is equal to 1. 

\textit{Modeling recommendations:}
We assume that the effect of recommendations depends on a normalized popularity index $\pi^{(i)}(t)$, which is an aggregate measure of the attention given to previous influencer content. We thus choose to define the popularity index of each influencer $i\in\mathcal{I}$ as
\[
\pi^{(i)}(t) = {\sum_{v \in \mathcal{V}} x^{(i)}_v(t)}\bigg/{\sum_{j \in \mathcal{I}} \sum_{w \in \mathcal{V}} x^{(j)}_w(t)}.
\]
Using such a normalized measure of popularity accounts for constraints on the global potential audience, which are referred to as ``finite carrying capacity'' of the medium in classical works about public arenas~\cite{hilgartner1988rise} and have been successfully included in recent mathematical models of social media~\cite{castaldo2021junk}.

\textit{Modeling content quality:} 
We also suppose that intrinsic factors, such as the quality of the contents proposed by the influencers, act as persistent input in the attention dynamics: the quality affects users by the exposure to the contents.
The quality of contents posted by influencer $i$ will be denoted by $q^{(i)}\in[0,1]$, with the intuition that higher quality contents should receive more attention. 

The contributions of interactions, popularity, and quality are weighted by nonnegative coefficients $\alpha_v$, $\beta_v$ and $\gamma_v$, 
such that $
\alpha_v + \beta_v + \gamma_v = 1.
$
Therefore, we assume that  the dynamic of the attention $x_v^{(i)}(t)$ evolves for all $v\in\mathcal{V}$ according to:
\begin{equation}\label{eq:mod_iniziale_det}
x_v^{(i)}(t+1) = \alpha_v \sum_{w \in\mathcal{V}} P_{vw} x^{(i)}_w(t) + \beta_v \pi^{(i)}(t) + \gamma_v q^{(i)},
\end{equation}
which can be rewritten in vector form as
\begin{equation*}\label{eq:sist_completo_det_matrici}
x^{(i)}(t+1)=APx^{(i)}(t)+B\pi^{(i)}(t)+\Gamma q^{(i)},
\end{equation*}
where $A=\mathrm{diag(\alpha)}$, $B=\mathrm{diag(\beta)}$, and $\Gamma=\mathrm{diag(\gamma)}$.

\section{Qualitative and numerical analysis} \label{section:results}

In this section, we study the asymptotic behavior of the dynamics. We will proceed by first considering, in Section~\ref{section:alpha0}, the case of $\alpha_v=0$ for all $v \in \mathcal{V}$, that is, when there is no social influence.  
Next, Section~\ref{section: gamma0} is devoted to the case with $\gamma_v=0$ for all $v\in\mathcal{V}$, that is, when quality plays no role and only network effects and recommendations influence the dynamics.
Notice that in the case $\beta_v=0$ for all $v\in\mathcal{V}$, that is, when only the network effect and the quality play a role in the dynamics, the system reduces to a set of $n$ decoupled \textit{Friedkin-Johnsen} models~\cite{FriedJon}. 
Finally, in Section~\ref{section:tuttitermini}, we will study what happens in the general case when all terms give contribution to the dynamics. 

\subsection{Asymptotic Behavior Without Network Effects} \label{section:alpha0}
In the dynamics described by equation \eqref{eq:mod_iniziale_det}, each future attention value is influenced by user interactions in the network described by the weighted-adjacency matrix~$P$. 
Assuming $\alpha_v=0 \ \forall v\in\mathcal{V}$ eliminates the contribution of interactions, retaining only the influence from past popularity values and other factors. By taking into account that in this case $\gamma_v=1-\beta_v$, we can rewrite \eqref{eq:mod_iniziale_det} as follows: 
\begin{equation}\label{equation:alpha_equal_0}
x_v^{(i)}(t+1) =  \beta_v \pi^{(i)}(t) + (1-\beta_v) q^{(i)}.
\end{equation}

The following result guarantees that for any initial condition the dynamics in \eqref{equation:alpha_equal_0} converges asymptotically to a limit point.

\begin{theorem}[Convergence -- No network]\label{proposition:alpha0}
Let $\alpha_v=0$ for all $v\in\mathcal{V}$ and denote $q_{\mathrm{tot}}=\sum_{i\in\mathcal{I}}q^{(i)}$. If $\exists \ v\in\mathcal{V} $ s.t. $\beta_v\neq1$ and  $\exists i\in\mathcal{I}$ s.t. $q^{(i)}>0$, then, for any initial condition $\pi(0)=\pi_0$, the dynamics \eqref{equation:alpha_equal_0} is such that $\forall i\in\mathcal{I}$,  $\forall v\in\mathcal{V}$,
$$\lim_{t\rightarrow\infty}\pi^{(i)}(t)=\frac{q^{(i)}}{q_{\rm{tot}}},\ 
\lim_{t\rightarrow\infty}x^{(i)}_v(t)=\beta_v\frac{q^{(i)}}{q_{\rm{tot}}}+(1-\beta_v)q^{(i)}.
$$ 
\end{theorem}

\begin{proof}
By denoting $\overline{\beta}=\frac{1}{n}\sum_{v\in\mathcal{V}}\beta_v$  and $\overline{x}^{(i)}(t) = \frac{1}{n}{\sum_{v \in V} x^{(i)}_v(t)}$, from
\eqref{equation:alpha_equal_0}, we obtain
\[
x_v^{(i)}(t+1) = \beta_v \frac{\overline{x}^{(i)}(t)}{{\sum_{j \in \mathcal{I}}\overline{x}^{(j)}(t)}} + (1-\beta_v) q^{(i)},
\]
from which
$
\sum_{i \in \mathcal{I}}\overline{x}^{(i)}(t+1)= \overline{\beta} + (1-\overline{\beta}) q_{\mathrm{tot}}
$
and
\[
\pi^{(i)}(t+1) = \frac{\overline{\beta}}{\overline{\beta} + (1-\overline{\beta}) q_{\mathrm{tot}}} \pi^{(i)}(t) + \frac{(1-\overline{\beta})}{\overline{\beta} + (1-\overline{\beta}) q_{\mathrm{tot}}} q^{(i)}.
\]
Since $q_{\mathrm{tot}}>0$ and $\overline{\beta}<1$, this scalar dynamics is asymptotically stable and the statement follows immediately.
\end{proof}

We observe that, in this case, the steady state is fully determined by the intrinsic quality of the influencers.
Numerical simulations confirm the result, showing 
convergence for any initial condition. In Figure~\ref{fig:alpha}, we present a simple example. 
We consider three influencers and 20 users and we set  $q^{(i)}=(0.3,0.7,0.5)$ and the parameters $\beta_v$ sampled from a uniform distribution over $[0,1]$. The initial conditions $x_v^{(i)}(0)$ are also sampled from uniform distributions on $[0,1]$. 
\begin{figure}
\includegraphics[width=0.49\columnwidth,trim=10 0 10 0]{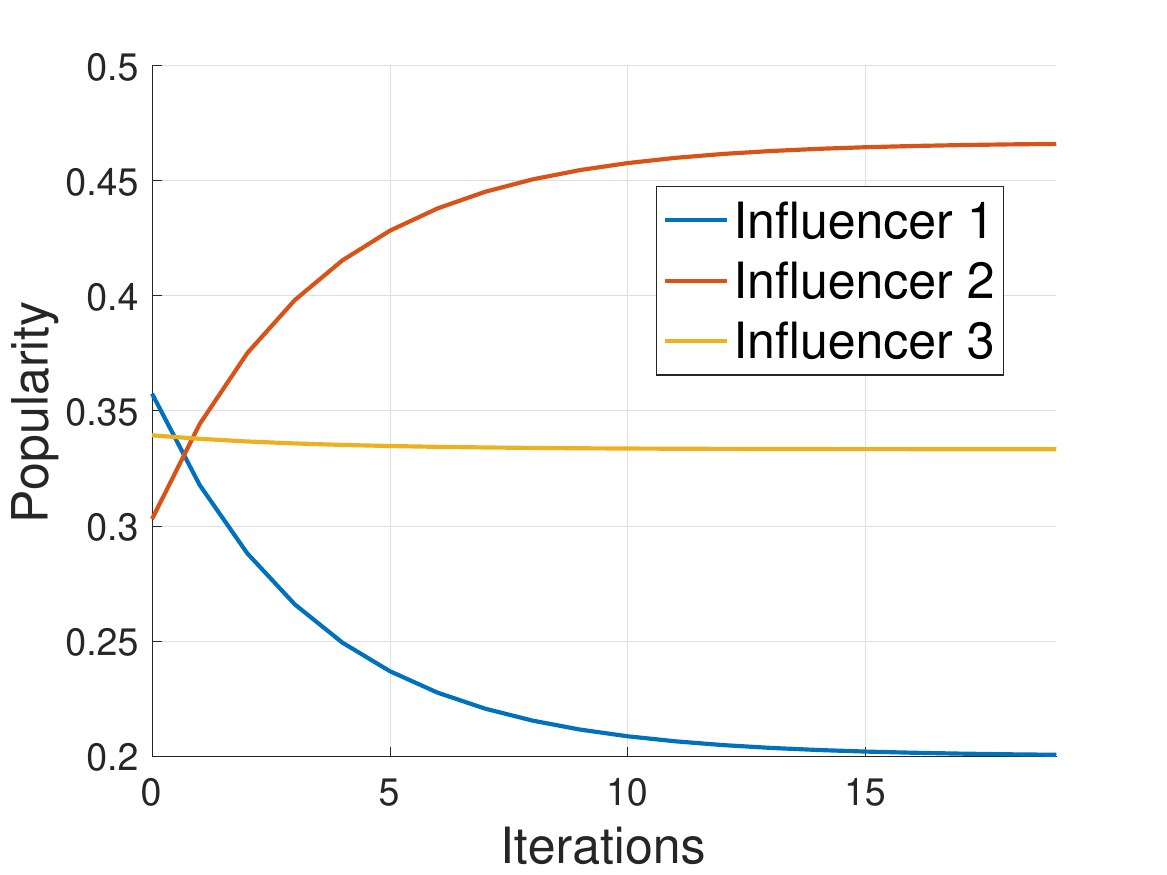}        \includegraphics[width=0.49\columnwidth,trim=10 0 10 0,clip]{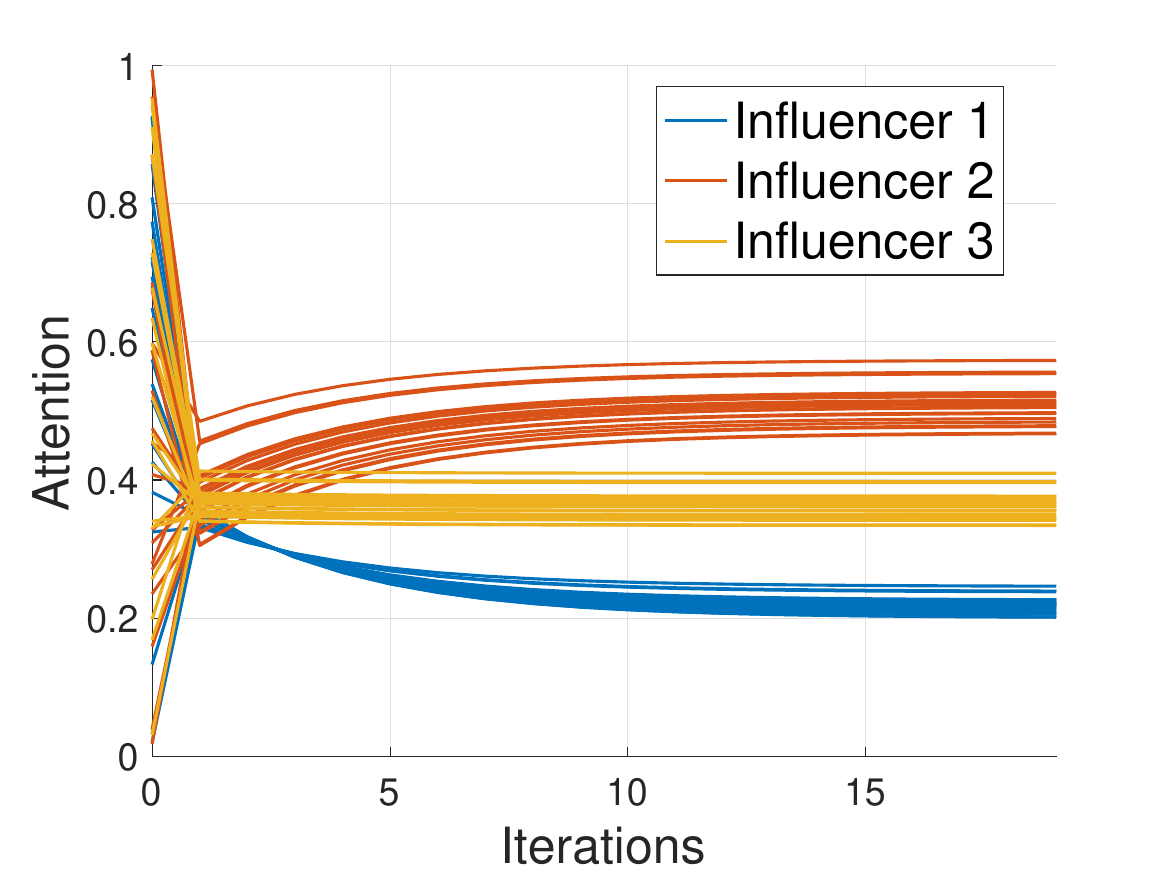}    
\caption{Evolution of dynamics~\eqref{equation:alpha_equal_0}: Popularity $\pi^{(i)}(t)$ (left) and attention $x_v^{(i)}(t)$ (right).} 
    \label{fig:alpha}
\end{figure}

\subsection{Asymptotic behavior without quality effects} \label{section: gamma0}
We now consider the case in which the quality gives no contribution, i.e. when $\gamma_v=0 \ \forall v \ \in \mathcal{V}
$.
Since $\beta_v=1-\alpha_v$, the dynamics in \eqref{eq:mod_iniziale_det} becomes 
  \begin{equation}
\label{equation:gamma0}
    x^{(i)}(t+1) = APx^{(i)}(t) + (I - A) \pi^{(i)}(t),
\end{equation}
where $A=\mathrm{diag}(\alpha)$.
Let us define $z(t)=\sum_{i\in\mathcal{I}} x^{(i)}(t)$ and use it to rewrite $\pi^{(i)}(t)$ in terms of $z(t)$ as
\begin{equation} \label{equation: pi}
\pi^{(i)}(t)=\frac{\sum_{v\in\mathcal{V}} x_v^{(i)}(t)}{\sum_{w \in \mathcal{V}}\sum_{j \in \mathcal{I}}x_w^{(j)}(t)} = \frac{\mathds{1}_n^{\top}x^{(i)}(t)}{\mathds{1}_n^{\top}z(t)}.
\end{equation}

\begin{proposition} \label{prop:zeta}
Let $\gamma_v=0$ for all $v \in \mathcal{V}$ and assume that, in the graph associated to $P$, there exists for all $v \in \mathcal{V}$ a path from $v$ to $w$ with $\alpha_w <1$. Then, $z(t)$ converges to $z^{\star} = \mathds{1}_n$.

\end{proposition}
\begin{proof}
  From equation \eqref{equation:gamma0}, recalling that $\sum_{i \in \mathcal{I}} \pi^{(i)}(t) = 1$ for all $t$ we obtain,
$
z(t+1)=APz(t)+(I-A)\mathds{1}_n.
$
This is a specific instance of a Friedkin-Johnsen model~\cite{FriedJon}, in which the constant input is  $\mathds{1}_n$. {\color{black} Then, the reachability of the deficient node $v$ implies, by \cite[Lemma~5]{FRASCA2013212}, that $AP$ is Schur stable  
and therefore $
  z^\star = (I-AP)^{-1}(I-A)\mathds{1}_n=\mathds{1}_n,
  $ where the last equality follows immediately from $P\mathds{1}_n=\mathds{1}_n.$}
\end{proof}

It is convenient to treat the evolution of $x^{(i)}(t)$ and of $\pi^{(i)}(t)$ jointly. 
By defining  
$s^{(i)}(t) := (x^{(i)}(t)^{\top} , \pi^{(i)}(t))^{\top} ,
$ the joint model becomes
\begin{equation} \label{definition:model_s}
    s^{(i)}(t+1)=U(t)s^{(i)}(t),
\end{equation}
with
\begin{equation} \label{definition: U}
U(t) = \begin{bmatrix} 
AP & (I-A)\mathds{1}_n \\ 
\frac{\mathds{1}_n^TAP}{\mathds{1}_n^{\top}z(t)} & \frac{\mathds{1}_n^{\top}(I-A)\mathds{1}_n}{\mathds{1}_n^{\top}z(t)}
\end{bmatrix}.
\end{equation}
Notice that $\mathds{1}_n^{\top}z(t)$ converges to $n$ by Proposition~\ref{prop:zeta} and therefore $U(t)$ converges to the row stochastic matrix 
\begin{equation}\label{eq:first-Utilde}
    \widetilde{U} = \begin{bmatrix} 
AP & (I-A)\mathds{1}_n \\ 
\frac{\mathds{1}_n^{\top} AP}{n} & \frac{\mathds{1}_n^{\top} (I-A)\mathds{1}_n}{n} 
\end{bmatrix}.
\end{equation}
We therefore have the following result, whose proof is postponed to Appendix \ref{app:A}.
\begin{theorem}[Convergence -- No quality]
    \label{proposition:convergence_s}
    Let $\gamma_v=0$ for all $v \in \mathcal{V}$ and {\color{black}{$z(0)\geq\1$}}. Assume that in the graph associated to $P$, for all $v \in \mathcal{V}$, there exists a path from $v$ to an {\color{black}{aperiodic node $w$}} with $\alpha_w <1$. Then, for all $i \in \mathcal{I}$ the dynamics \eqref{definition:model_s} converges to 
    $
    s^{(i)\star}= \phi^{\top} s^{(i)}(0) \, \mathds{1}_{n+1}$ as $t \rightarrow \infty
    $  
    for some vector $\phi\in \mathbb{R}^{n+1}$.
Moreover, let $\lambda_1=\rho(AP)$ and let $\tilde\phi$ be the stationary distribution of $\widetilde{U}$, then there exists a constant $\chi>0$
   {\color{black} $$
    \|{\phi}-\tilde{\phi}\|_1\leq \chi\|z(0)-\mathds{1}\|_1\frac{\lambda_1\mathfrak{p}_n(\lambda_1)}{(1-\lambda_1)^{n+1}}.
$$
where $\mathfrak{p}_n(\lambda_1)$ is a polynomial of degree $n$ and $\mathfrak{p}_n(0)=1$.}
\end{theorem} 

It can be seen that in this case, the dynamics for each influencer $i$ converges to a consensus, since the role of each user is the same: in other words, all users converge to the same value of attention for the $i$-th influencer, which is equal to its popularity. Although $\phi$ is the same for all $i$, the consensus value is different depending on the influencer $i$, since it depends on the initial condition $s^{(i)}(0)$.  
{\color{black}{Moreover, for fixed $n$, when the influence of the matrix $AP$ is weak, $\lambda_1$ will be small and $\phi$ will approach the reference stationary distribution $\tilde{\phi}$.
On the other hand, when $\lambda_1$ approaches 1, the bound will be less informative, but we expect that $\phi$ will be close to the invariant measure of $P$.}}

Figure~\ref{figure:gamma0} shows an example. The network is modeled as an Erdős-Rényi random graph with $n=20$ nodes and edge probability $ p = 0.2 $. The parameters $ \beta_v $ and initial conditions $x_v^{(i)}(0)$ are sampled from a uniform distribution in $[0,1]$. %the evolution of popularity and appreciation of the users, iterating the updates for 20 instants of time.
It can be observed that, for any influencer $i$, the attention values by all users $x_v^{(i)}(t)$ go to a consensus, as deduced in theory. The consensus value is the limit popularity of $i$, $\lim_{t\rightarrow \infty}\pi^{(i)}(t)$.
Simulations also confirm that $\phi$ is close to $\tilde{\phi}$:
the dashed lines in the left plot of Figure~\ref{figure:gamma0} show the consensus values that would be produced by using $\tilde{\phi}$ instead of $\phi$.
\begin{figure}%[ht]
    \centering    \includegraphics[width=0.49\columnwidth,trim=10 0 10 0]{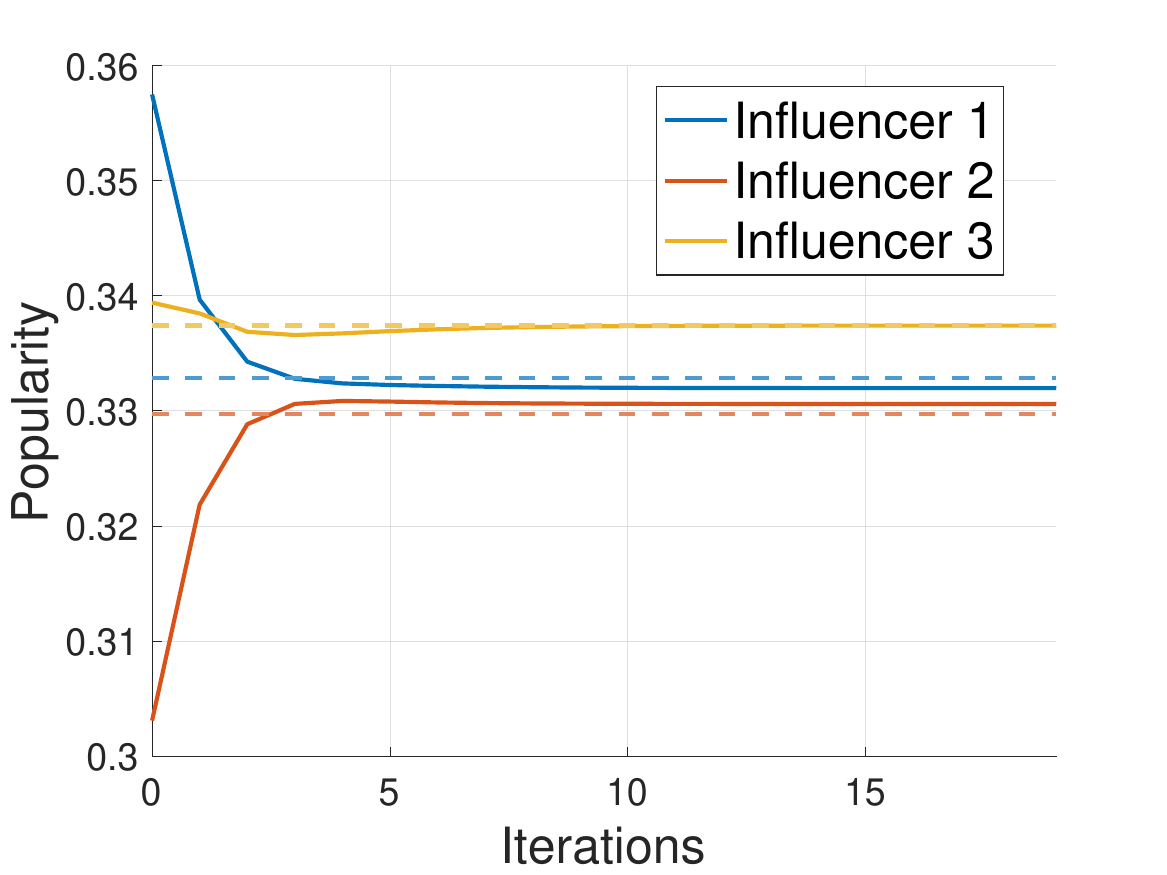} 
\includegraphics[width=0.49\columnwidth,trim=10 0 10 0]{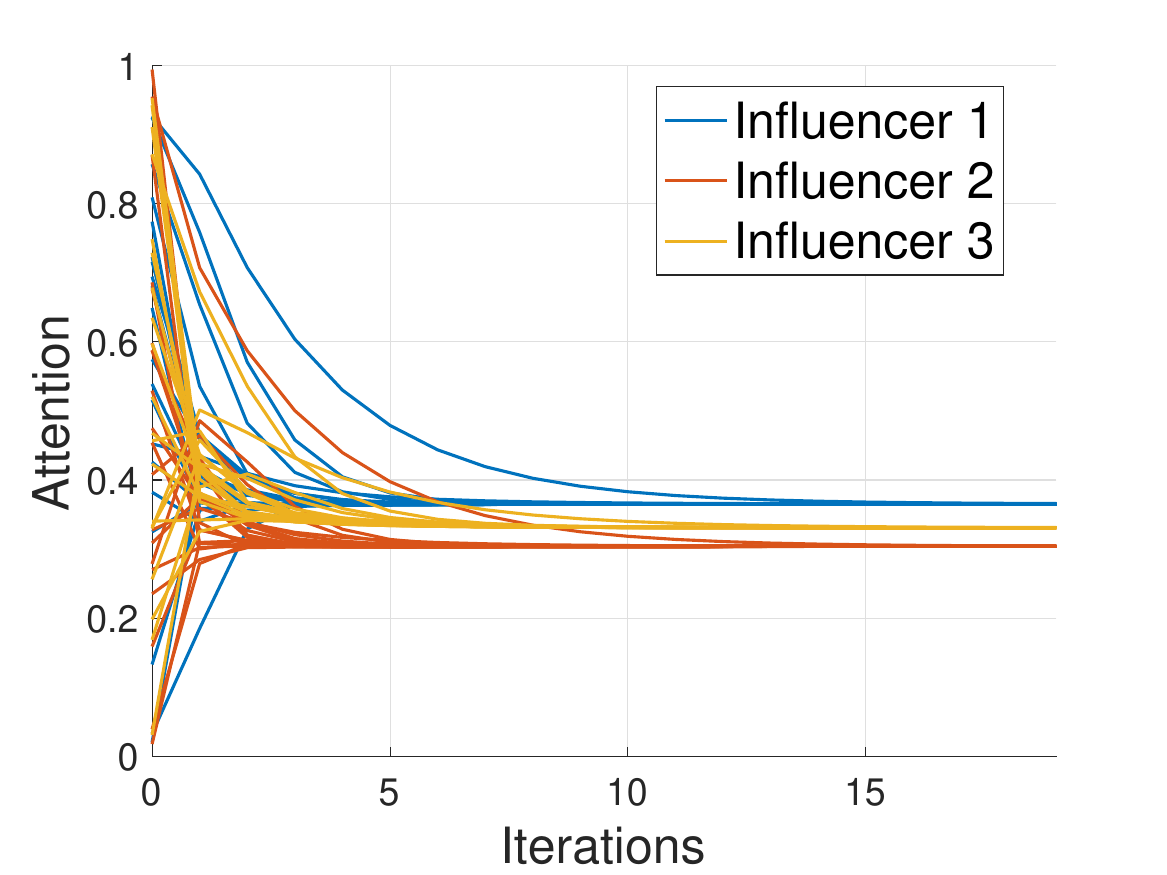}      \caption{Evolution of dynamics~\eqref{equation:gamma0}: Popularity $\pi^{(i)}(t)$ (left) and attention $x_v^{(i)}(t)$ (right).} 
    \label{figure:gamma0}
\end{figure}

\subsection{Asymptotic behavior with all effects} \label{section:tuttitermini}
It remains to understand how the general dynamics, in which all terms appear, behave.
We define 
$s^{(i)}(t) := (x^{(i)}(t)^{\top} , \pi^{(i)}(t))^{\top} ,
$ and express the coupled dynamics
\begin{equation} \label{sist_completo}
s^{(i)}(t+1)=U(t)s^{(i)}(t)+{q}^{(i)}c(t)
\end{equation}
where 
{\begin{equation}\label{expressions:Uc}
U(t) = \begin{bmatrix} 
AP & B\mathds{1}_n \\ 
\frac{\mathds{1}_n^{\top}AP}{\mathds{1}_n^{\top}z(t)} & \frac{\mathds{1}_n^{\top}B\mathds{1}_n}{\mathds{1}_n^{\top}z(t)}
\end{bmatrix}, \quad 
c(t) = \begin{bmatrix} 
(I-A-B)\mathds{1}_n \\ 
\frac{\mathds{1}_n^{\top}(I-A-B)\mathds{1}_n}{\mathds{1}_n^{\top}z(t)} 
\end{bmatrix}.
\end{equation}

Notice that also in this case we have the convergence of $z$.
\begin{proposition} \label{prop:zeta2}
Assume that for all $v\in \mathcal{V}$ there exists a path from $v$ to $w$ such that $\gamma_w>0$ in the graph associated to $P$. Then $z(t)$ converges and
\[
z^{\star}:=\lim_{t\rightarrow\infty}z(t)= (I-AP)^{-1}(B +q_{\mathrm{tot}}(I-A-B))\mathds{1}_n.
\]
\end{proposition}
\begin{proof} 
Notice that in this case \[
z(t+1)=AP z(t)+ B \mathds{1}_n+q_{\mathrm{tot}}(I-A-B) \mathds{1}_n.
\]
As in the proof of Proposition~\ref{prop:zeta}, the connectivity assumption guarantees that $AP$ is Schur stable by Lemma 5 in \cite{FRASCA2013212}.
\end{proof}

The following theorem guarantees asymptotic convergence of $s$: its somewhat lengthier proof is given in Appendix~\ref{app:C}.
\begin{theorem}[Convergence -- Generic case]\label{theorem:tutti_termini}
    Assume that, in the graph associated to $P$, for all $v\in \mathcal{V}$ there exists a path from $v$ to {\color{black}{an aperiodic node $w$}} such that $\gamma_w>0$. 
        If $q_{\mathrm{tot}} \geq 1$ and $z(0)\geq \1_n$, then 
    the dynamics \eqref{sist_completo} 
    is convergent for $t\rightarrow \infty$  and  for all $i\in\mathcal{I}$, 
    $
\lim_{t\rightarrow\infty}s^{(i)}(t)=q^{(i)} \,(I-\widetilde{U})^{-1}\widetilde{c} ,
    $
    with
{\small{  \begin{equation}\label{eq:expressions_tilde}
\widetilde{U}=\begin{bmatrix} 
AP & B\mathds{1}_n \\ 
\frac{\mathds{1}_n^{\top} AP}{\mathds{1}_n^{\top}z^\star} & \frac{\mathds{1}_n^{\top} B\mathds{1}_n}{\mathds{1}_n^{\top}z^\star} 
\end{bmatrix}, \quad \widetilde{c}=\begin{bmatrix} (I-A-B)\mathds{1}_n \\ \frac{\mathds{1}_n^{\top}(I-A-B)\mathds{1}_n}{\mathds{1}_n^{\top}z^\star}  \end{bmatrix}. 
\end{equation}}}
    \end{theorem}
 
\smallskip This convergence result requires that $q_{\mathrm{tot}} \geq 1$ and $z_v(0)\ge 1$ for all $v \in \mathcal{V}$ {\color{black}(The latter condition also appears in Theorem~\ref{proposition:convergence_s})}. We argue that these conditions are mild and, in fact, not necessary for convergence.
Indeed, they require both the total quality of the influencers and each user's total attention for the influencers to be at least 1. This assumption is not restrictive as long as the number of competing influencers is not too small.
Furthermore, simulations show that the dynamics converge also when these conditions are not met. 
An example is provided in Figure~\ref{fig:casoparticolare}, where the network is the same as in Figure~\ref{figure:gamma0}. 
The parameters $\alpha_v, \beta_v,\gamma_v $ and $x_v^{(i)}(0)$ are sampled from a uniform distribution in $[0,1]$, then normalized to satisfy $\alpha_v+\beta_v+\gamma_v=1$ for all $v\in\mathcal{V}$.
\begin{figure}
\includegraphics[width=0.49\columnwidth,trim=10 0 10 0]{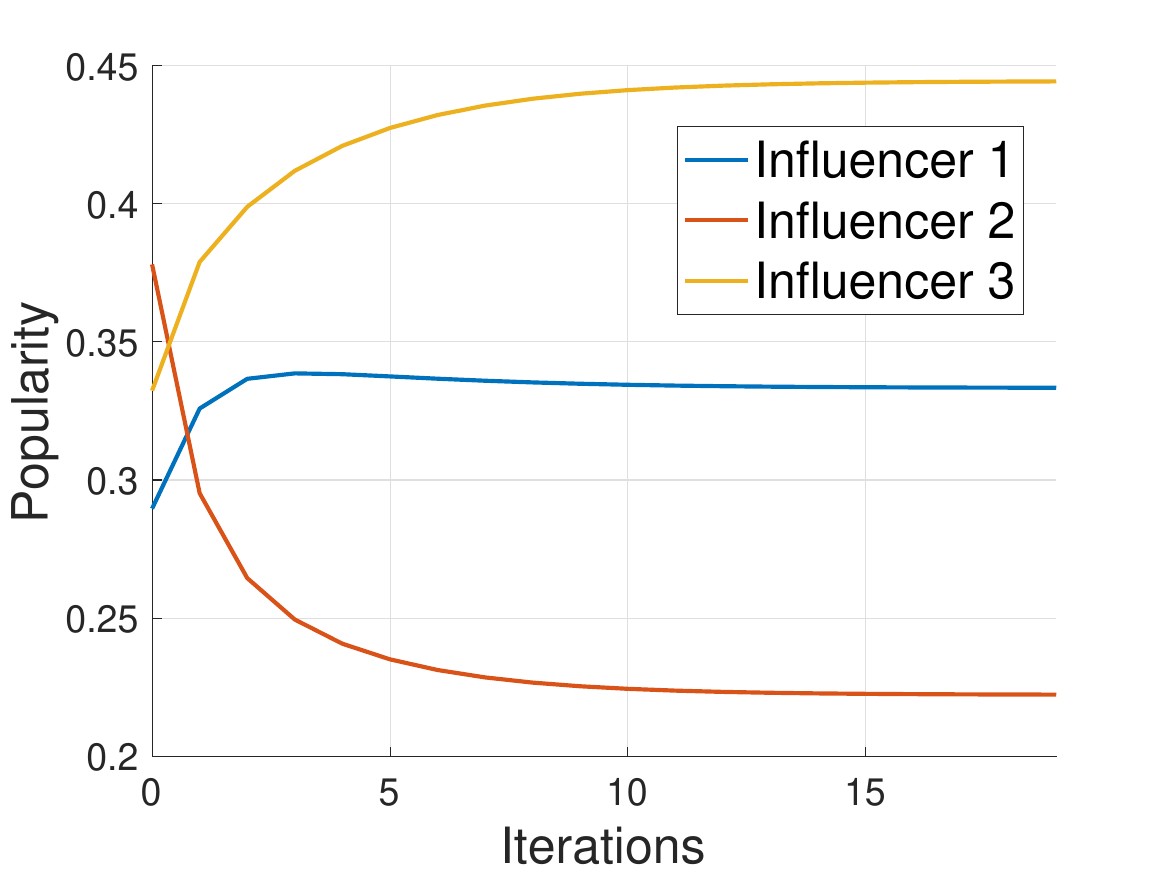}
\includegraphics[width=0.49\columnwidth,trim=10 0 10 0]{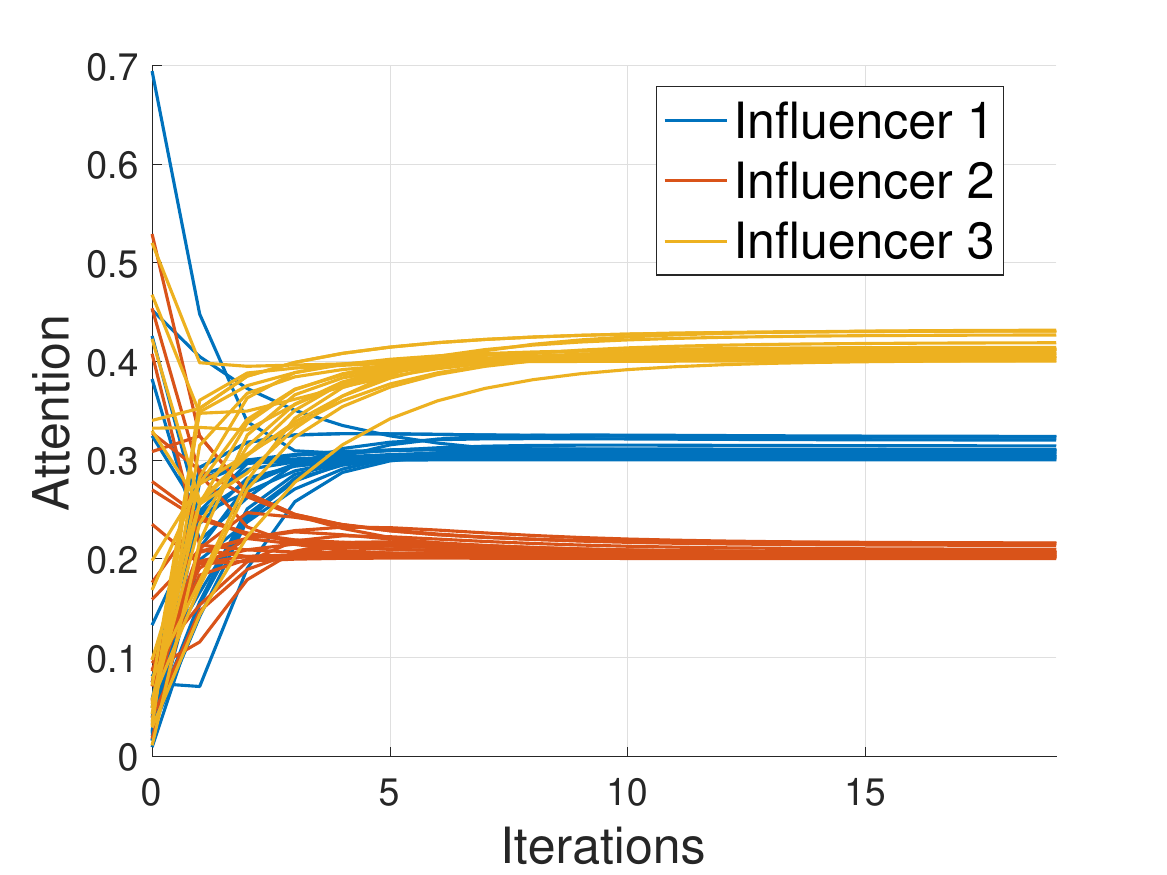}
  \caption{Evolution of popularity $\pi^{(i)}(t)$ (left) and attention $x_v^{(i)}(t)$ (right) in the general case~\eqref{sist_completo}, when $\alpha_v\neq0, \ \beta_v\neq0, \ \gamma_v\neq0$ for all $v$. Notice that condition $z_v(0)\ge 1$ in Theorem~\ref{theorem:tutti_termini} is not satisfied in this example. 
  }
        \label{fig:casoparticolare}
\end{figure}

\section{Discussion}\label{sect:discussion}
The study of the model shows that, under mild connectivity conditions, the dynamics converge to equilibria, which we are able to characterize explicitly. 

The characterizations of the equilibria allow us to deduce important insights on the relationship between social influence, recommendations, and content quality. 
Let us discuss the various cases of interest.

When the network has no role ($\alpha=0$) as in Theorem~\ref{proposition:alpha0}, popularity is determined solely by quality. 
Therefore, it appears that the network is essential for recommendations to have an effect. 
When recommendations (and thus past popularity) have no role {\color{black}($\beta=0$), the system reduces to $|\mathcal{I}|$ decoupled Friedkin-Johnsen's models with constant input $q^{(i)}\1_n$, whose equilibria read
${x^{(i)}}^*=(I-AP)^{-1}(I-A)q^{(i)}\mathds{1}_n=q^{(i)}\mathds{1}_n$
and are globally asymptotically stable, as long as the social network contains a globally reachable node. In this case, the popularity of the influencers} is trivially determined by their quality, thus making the network irrelevant. These two cases, taken together, indicate that recommendations enhance and enable the effects of the network. 
If quality has no role ($\gamma=0$) as in Theorem~\ref{proposition:convergence_s}, popularity is solely determined by the structure of the user network and by the initial condition: influencers who initially receive attention by prominent users will continue to be successful. In this scenario, the success of influencers is uniform across users (all users share the same taste) and if the alpha parameter is small, then influence becomes evenly distributed, thus fixing initial success. 
On the contrary, in the generic case with all effects, as in Theorem~\ref{theorem:tutti_termini}, quality plays a fundamental role. The steady-state vector of the attention received by an influencer is the product of its quality and an influence profile vector. Each user heeds influencers differently, but the influence profile is the same for every influencer and depends on the parameters $\gamma_v$'s and on the network structure, with no role for the initial condition. 

In conclusion, our analysis offers {\color{black}insights into the potential relationships} between quality, network effects, and popularity.
One key insight is that the network has a significant impact despite quality being a property of the influencer only, independent from the users. {\color{black}Indeed, the discussion on the cases with $\alpha=0$ and $\beta=0$ shows that popularity-based recommendations enhance the role of the social network.}
Another key insight is that {\color{black}the dynamics displays rather different regimes, depending on whether users are sensitive to quality}. If $\gamma$ is positive, quality plays a dominant role.
If $\gamma$ is zero and the initial conditions are random, their randomness plays a dominant role, making popularity more unpredictable.
This diversity of regimes may be useful to explain variations across different platforms, which can display forms of virality that are more or less predictable and more or less directly explainable by content features.

\section{\color{black} Conclusion}\label{sect:conclusion}
In this paper, we have formulated a nonlinear model that aims to explain how users' attention to influencers, and thus the popularity of the latter, evolves. {\color{black}Our original model, which generalizes the well-known linear model by Friedkin and Johnsen~\cite{FriedJon}, includes} three key factors: social influence through user interactions, popularity-based recommendations by the platform, and the influencer quality.
{\color{black} As we have discussed in the previous section, the regimes captured by our model (see Sections~III.B and III.C) are consistent with qualitative observations reported in recent studies: users of platforms such as TikTok can show a lower sensitivity to content quality and reward influencers regardless of their experience \cite{10.1145/3501247.3531551}, while for example on YouTube, a smaller number of creators tend to dominate attention thanks to the production of high-quality content \cite{CHEUNG2022102940}. Future research should seek more explicit empirical confirmations of these insights.} 

Further work would also be beneficial to overcome some limitations of our analysis (essentially, the assumption on $z_v(0)$ in Theorems~\ref{proposition:convergence_s} and~\ref{theorem:tutti_termini} {\color{black}and the connectivity assumptions on the social influence network}) and to extend the model to allow for distinct personal preferences of individuals $q^{(i)}_v$.

\bibliographystyle{IEEEtran}
\bibliography{references2.bib}

\appendix
\section{Proof of Theorem~\ref{proposition:convergence_s}}\label{app:A}
{\color{black} \begin{lemma}\label{lemma:norm_inf}
Let $ M \in \mathbb{C}^{n \times n} $ be a matrix with spectral radius $ \rho\in(0,1) $. Then 
$
\|M^k\|_{\infty} \leq C \, k^n\, \rho^k \quad \text{for all } k \in \mathbb{N},
$
where $ \chi >0$ is a constant (independent of $ k $ and $\rho$). 
\end{lemma}
 \begin{proof}
Let \( q $ be the number of distinct eigenvalues of the matrix \( M \in \mathbb{C}^{n \times n} $, denoted by \( \{\lambda_{\ell}\}_{\ell=1}^{q} $, and consider its Jordan canonical decomposition: $M = U J U^{-1}.$
Then,
$
\|M^k\|_{\infty} \leq \|U\|_{\infty} \|J^k\|_{\infty} \|U^{-1}\|_{\infty}.
$
The \( k \)-th power of a Jordan block of size \( s \) associated with eigenvalue \( \lambda \) has entries bounded by \( \binom{k}{m} |\lambda|^{k-m} \), for \( m = 0, \dots, s-1 \). Hence,
$
\|J^k\|_{\infty} \leq \max_\ell |\lambda_\ell|^k \sum_{m=0}^{s_\ell - 1} \binom{k}{m} |\lambda_\ell|^{-m},
$
where \( s_\ell $ is the size of the largest Jordan block associated with \( \lambda_\ell $.
Using the bound \( \binom{k}{m} \leq k^m \), we get for some constant \( \chi > 0 \), independent of \( k \),
$
\|J^k\|_{\infty} \leq \max_\ell |\lambda_\ell|^k k^{s_\ell - 1} \sum_{m=0}^{s_\ell - 1} |\lambda_\ell|^{-m} \leq \chi \rho^k k^n,
$
Therefore, there exists a constant \( C = \|U\|_{\infty} \|U^{-1}\|_{\infty} \chi $, independent of \( k $, such that
$\|M^k\|_{\infty} \leq C \rho^k k^n.
$
\end{proof}
}

{\color{black}
\begin{lemma}\label{lemma:geometric_series}
Let \( \lambda \in (0, 1) \) and \( n \in \mathbb{N} \). Then
$
\sum_{k=0}^{\infty} k^n \lambda^k = \frac{\lambda\mathfrak{p}_{n-1}(\lambda)}{(1 - \lambda)^{n+1}},
$
where \( \mathfrak{p}_{n-1}(\lambda) \) is a polynomial of degree $ n-1$ satisfying $\mathfrak{p}_{n-1}(0)=1$.
\end{lemma}
\begin{proof}
Proceed by induction on $n\in\mathbb{N}$. For \( n = 1 \) 
{\small{$$
\sum_{k=0}^{\infty} k\lambda^k =\lambda  \sum_{k=0}^{\infty} k\lambda^{k-1}=\lambda \frac{\mathrm{d}}{\mathrm{d}\lambda }(\frac{1}{1-\lambda})=\frac{\lambda}{(1 - \lambda)^2},
$$}}

\noindent which is in the desired form with $\mathfrak{p}_0(\lambda) = 1 $.
Assume by induction that the formula holds for $n - 1$,% i.e.,
{\small{$$\sum_{k=0}^{\infty} k^n \lambda^k = \lambda \frac{d}{d\lambda} \left( \sum_{k=0}^{\infty} k^{n-1} \lambda^k \right)
= \lambda \frac{d}{d\lambda} \left( \frac{\lambda\mathfrak{p}_{n-2}(\lambda)}{(1 - \lambda)^n} \right).
$$}}

\noindent By differentiating we obtain
{\small{\[
 \frac{d}{d\lambda} \!\!\left( \frac{\lambda\mathfrak{p}_{n-2}(\lambda)}{(1 - \lambda)^n} \right)
\!=\! \frac{(\mathfrak{p}_{n-2}(\lambda)+\lambda \mathfrak{p}'_{n-2}(\lambda) )(1-\lambda)+ n \lambda\mathfrak{p}_{n-2}(\lambda)}{(1 - \lambda)^{n+1}}.
\]}}

\noindent where the numerator is a polynomial of degree at most $n-1$, say $\mathfrak{p}_{n-1}$, with the property $\mathfrak{p}_{n-1}(0)=\mathfrak{p}_{n-2}(0)=1.$
Multiplying by \( \lambda \), we get the result.
\end{proof}
}
{\color{black}
\begin{lemma} \label{lemma:U_st}
If $z(0)\geq \mathds{1} $, then $\|U(t)\|_{\infty}=1$, $\forall t\in\mathbb{N}$.
\end{lemma}
\begin{proof}
Since 
$P$ is row-stochastic and $z(0)\geq\1$, one can show by induction on $t\in\N$ that $z(t)\geq\1$ for all $t\in\N$. This implies that the sum of last row is not larger than 1. Moreover, the sum of the entries in each of the remaining rows is exactly 1. Therefore $\|U(t)\|_{\infty}=1$ for all $t\in\N$.
\end{proof}}
{\color{black}We will use of the following known result, which holds for any submultiplicative matrix norm.}}
\begin{lemma}[Theorem 4.1 in \cite{Stewart2003}] \label{lemma:gamma0}
{\color{black} Let the square matrix $Q$ have a simple dominant eigenvalue $\lambda_1(Q)=1$ with multiplicity $1$ and let $\widetilde{x}$ be the right eigenvector associated to this eigenvalue. Let $\Delta Q(k)$, for $k\in \N$, be a sequence of matrices of suitable dimension and define}

{\small{\begin{equation} \label{def:Product}
R(t) := \prod_{\ell=1}^t(Q+\Delta Q(\ell)).
\end{equation}}}
    If $
    \sum_{s=1}^\infty \|\Delta Q(s)\|<\infty,
    $
    then
   $
    \lim_{t\rightarrow\infty} R(t) = \widetilde{x}\phi^{\top}
    $
    for a suitable vector $\phi$.
    The exponential rate of convergence is not greater than
    $
    \max\{\rho, \sigma\},
    $
    where $\rho $ is the magnitude of the largest-in-magnitude of the subdominant eigenvalues of $\widetilde{U}$ and
    $
    \sigma = \limsup_t \|\Delta Q(t)\| ^{1/t}.
    $
    \hfill$\square$
\end{lemma}
\medskip

We let $U(t)=\widetilde{U}+\Delta U(t)$ with
{\small{\begin{align} \label{eq:decomposition} \Delta U(t)&={\textstyle \left( \frac{1}{\mathds{1}_n^{\top}z(t)} - \frac{1}{n}\right)}\begin{bmatrix} 
0_n & 0 \\ 
\mathds{1}_n^{\top}AP & \mathds{1}_n^{\top}(I-A)\mathds{1}_n
\end{bmatrix}.\end{align}}}

\noindent Notice that
$
\|\Delta U(t)\|_{\color{black}\infty}   \leq\left|\frac{1}{\mathds{1}_n^{\top}z(t)} - \frac{1}{n} \right| {\color{black}n}.
$
{\color{black} The connectivity and aperiodicity assumptions imply that $AP$ is Schur stable with spectral radius $0<\lambda_1<1$. 
Then
{\small{\begin{align*}
&\left| \frac{1}{\mathds{1}_n^{\top}z(t)} - \frac{1}{n} \right|=\frac{\left| \mathds{1}_n^{\top}(z(t)-\mathds{1}_n) \right|}{n\,\mathds{1}_n^{\top}z(t)}\leq\frac{\|(AP)^t\|_1\|z(0)-\mathds{1}_n\|_1}{n^2}\\
&\quad\leq \frac{\|(AP)^t\|_{\infty}\|z(0)-\mathds{1}_n\|_1}{n}\leq \chi\frac{\|z(0)-\mathds{1}_n\|_1\lambda_1^tt^n}{n}
\end{align*}}}
\noindent where the last inequality follows from Lemma \ref{lemma:norm_inf}. Since 
$
\sum_{\ell=1}^{\infty}\|\Delta U(\ell))\|_{\infty} <\infty, 
$}
the convergence of the dynamics is then a direct application of Lemma~\ref{lemma:gamma0}.

Iterating the decomposition of matrix $U(k)=\widetilde{U}+\Delta U(k)$ and using the triangular inequality we have
{\small{\begin{align*}&\|\prod_{k=0}^{t}U(k)-\widetilde{U}^{t+1}\|_{\color{black}\infty}\leq\sum_{k=0}^t\|\widetilde{U}^{k}\|_{\color{black}\infty}\|\Delta U(t-k)\|_{\color{black}\infty}\!\!\!\!\prod_{\ell=0}^{t-k-1}\!\!\!\|U(\ell)\|_{\color{black}\infty}.
    \end{align*}}}
    Since {\color{black}{$\|\widetilde{U}^\ell\|_{\infty}=\|U(\ell)\|_{\infty}=1$  for all $\ell\in\mathbb{N}$ (see Lemma \ref{lemma:U_st})}} and applying Lemma \ref{lemma:geometric_series} we get
   {\color{black}\small{ \begin{align*}
&\|\prod_{k=0}^{t}U(k)-\widetilde{U}^{t+1}\|_{\infty}\leq\chi\|z(0)-\mathds{1}\|\sum_{k=0}^t k^n\lambda_1^k\\
&\qquad\leq \chi\|z(0)-\mathds{1}\|\frac{\lambda_1\mathfrak{p}_n(\lambda_1)}{(1-\lambda_1)^{n+1}}    ,    
    \end{align*}
    where $\mathfrak{p}_n(\lambda_1)$ is a polynomial of degree equal to $n$ with $\mathfrak{p}_n(0)=1$.}} 
    By Perron-Frobenius theorem, the matrix power~$\widetilde{U}^{t}$ converges to $ \mathds{1}_{n+1}\tilde{\phi}^{\top}$ as $t\to\infty$.
    {\color{black}{Then, using triangular inequality and Lemma \ref{lemma:gamma0}, there exists a constant $C>0$ such that
    {\small{\begin{align*}\|{\phi}-\tilde{\phi}\|_1&\leq\|\mathds{1}_{n+1}{\phi}^{\top}-\prod_{k=0}^{t}U(k)\|_{\infty}\\
    & +\|\prod_{k=0}^{t}U(k)-\widetilde{U}^{t+1}\|_{\infty}+\|\widetilde{U}^{t+1}-\mathds{1}_{n+1}{\tilde{\phi}}^{\top}\|_{\infty}\\
    &\leq 
     \chi\|z(0)-\mathds{1}\|\frac{\lambda_1\mathfrak{p}_n(\lambda_1)}{(1-\lambda_1)^{n+1}}  +C \max\{\lambda_2(\widetilde U), \lambda_1\}^t,
    \end{align*}}}}}
where $\lambda_2(\widetilde U)<1$ is the largest of magnitudes of the subdominant eigenvalues of $\widetilde{U}$ (see Lemma~\ref{lemma:gamma0}).
    The statement is then obtained by letting $t$ go to $\infty.$

\section{Proof of Theorem~\ref{theorem:tutti_termini}}\label{app:C}
\begin{lemma} \label{lemma:U_subst}
Assume that for all $v\in \mathcal{V}$ there exists, in the graph associated to $P$, a path from $v$ to $w$ such that $\gamma_w>0$. If
$q_{\mathrm{tot}} \geq 1$ and $z(0)\geq \mathds{1} $, then $U(t)$ is Schur stable $\forall t\in\mathbb{N}$.
\end{lemma}
\begin{proof}
First, using stochasticity of $P$, it can be proved that $z(t)\geq \mathds{1}_n$ by induction on $t\in\mathbb{N}_0$. 
Consequently, the sum of the last row (and of all other rows) in $U(t)$ is not larger than 1, because
$\mathds{1}_n^{\top}AP+\mathds{1}_n^{\top}B\mathds{1}_n\leq\sum_i\alpha_i + \sum_i \beta_i\leq n\leq\mathds{1}_n^{\top}z(t).  
$ 
Finally, the connectivity assumption implies that every node is connected to %a deficiency node, i.e.\ 
a node whose corresponding row in the matrix $U(t)$ sums to less than one.  Applying Lemma 5 in \cite{FRASCA2013212}, we conclude that $U(t)$ is Schur stable. 
\end{proof}
\smallskip
We consider now the joint dynamics in \eqref{sist_completo}.
The dynamics in \eqref{sist_completo} can be written in the following form
{\small{\begin{equation} \label{eq: sist con i termini}
s^{(i)}(t+1)=\left( \prod_{k=0}^{t}U(k)\right) s^{(i)}(0) + q^{(i)} \sum_{k=0}^{{t}} \left( \prod_{j=k+1}^{t} U(j)\right) c(k).
\end{equation}}}
where $U(t)=\widetilde U+\Delta U(t)$ with $\widetilde{U}$ given in \eqref{eq:expressions_tilde} and
\begin{gather}\label{eq:dec}
 \Delta U(t)={\textstyle \left( \frac{1}{\mathds{1}_n^{\top}z(t)} - \frac{1}{\mathds{1}_n^{\top}z^*}\right)}\begin{bmatrix} 
0_n & 0 \\ 
\mathds{1}_n^{\top}AP & \mathds{1}_n^{\top}B\mathds{1}_n
\end{bmatrix},
\end{gather}
and $c(t)=\widetilde{c}+\Delta c(t)$ with $\widetilde{c}$ given in \eqref{eq:expressions_tilde}
\begin{gather}\label{eq:def_c}
\Delta c(t)= \begin{bmatrix} 0 \\ \mathds{1}_n^{\top}(I-A-B)\mathds{1}_n(\frac{1}{\mathds{1}_n^{\top}z(t)}-\frac{1}{\mathds{1}_n^{\top}z^*})
\end{bmatrix}. 
\end{gather}
{\color{black}From now on, we will denote the infinity norm by $\|\cdot\|$.
}
\begin{remark}\label{remark:bounds} By Proposition~\ref{prop:zeta2}, there exist $\chi_1$, $\chi_2>0$ such that $\|\Delta c(k)\|\leq \chi_1\|(AP)^k\|$ and $\|\Delta U(k)\|\leq \chi_2\|(AP)^k\|$.
\end{remark}
\begin{lemma}\label{lemma:decomposition}
The matrices defined in \eqref{expressions:Uc} satisfy
{\small{\begin{align*}
&\sum_{k=0}^t\prod_{s=k+1}^tU(s)c(k)=\sum_{k=0}^t\widetilde{U}^{t-k}\widetilde{c}+\sum_{k=0}^t\widetilde{U}^{t-k}\Delta{c}(k)\\
&\qquad+\sum_{k=0}^t\sum_{\ell=0}^{t-k-1}\widetilde{U}^{\ell}\Delta U(t-\ell)\prod_{s=k+1}^{t-\ell-1}U(s)c(k).
\end{align*}}}
\end{lemma}
\begin{proof}
By repeated application of $U(t)=U+\Delta U(t)$,
{\small{\begin{align*}
&\prod_{s=k+1}^tU(s)=(\widetilde U+\Delta U(t))\prod_{s=k+1}^{t-1}U(s)\\
&\quad=\widetilde U^{t-k}+\sum_{\ell=0}^{t-k-1}\widetilde{U}^{\ell}\Delta U(t-\ell)\prod_{s=k+1}^{t-\ell-1}U(s).\end{align*}}}
We conclude by multiplying the expression by $c(k)$, summing over $k\in\{0,\ldots,t\}$, and applying $c(k)=\widetilde{c}+\Delta c(k)$.
\end{proof}

Starting from \eqref{eq: sist con i termini} and the definition of $\widetilde{s}(t+1)$, we can apply Lemma~\ref{lemma:decomposition} to show that
{\small{\begin{align*}
&\|s^{(i)}(t+1)-\widetilde{s}^{\color{black}{(i)\star}}\|\\
&\leq\|\prod_{k=0}^tU(k)\||\|s^{(i)}(0)\| +q^{(i)}\sum_{k=0}^t\|\widetilde{U}^{t-k}\|\|\Delta{c}(k)\|\\&\quad+q^{(i)}\sum_{k=0}^t\sum_{\ell=0}^{t-k-1}\|\widetilde{U}^{\ell}\|\left\|\Delta U(t-\ell)\right\|\prod_{s=k+1}^{t-\ell-1}\left\|U(s)\right\|\left\|c(k)\right\|\\
&\leq \|\widetilde{U}^{t+1}\|\|s^{(i)}(0)\|+\sum_{k=0}^t\sum_{\ell=0}^{t-k-1}\|\widetilde{U}^{\ell}\|\left\|\Delta U(t-\ell)\right\|\left\|s^{i}(0)\right\|\\
&\quad+q^{(i)}\sum_{k=0}^t\|\widetilde{U}^{t-k}\|\|\Delta{c}(k)\|\\
&\quad+q^{(i)}\sum_{k=0}^t\sum_{\ell=0}^{t-k-1}\|\widetilde{U}^{\ell}\|\left\|\Delta U(t-\ell)\right\|\left\|c(k)\right\|
\end{align*}}}
where the last inequality is obtained from Lemma~\ref{lemma:decomposition}, by triangular inequality and observing $\prod_{s=k+1}^{t-\ell-1}\left\|U(s)\right\|\leq1$. {\color{black}{
{\color{black}{The connectivity and aperiodicity assumptions imply that both AP and $\widetilde{U}$ have spectral radius in (0,1).}}
By Remark~\ref{remark:bounds}, Lemma \ref{lemma:norm_inf} and using the boundness of $c(k)$,  there exist positive constants $\kappa_1,\kappa_2$, and $\kappa_3$ such that  
{\small{\begin{align}\label{eq:last_expression}
&\|s^{(i)}(t+1)-\widetilde{s}^{(i)\star}\nonumber\|\\
&\quad\leq\kappa_1t^n\rho(\widetilde{U}) ^{t+1}+\kappa_2\sum_{k=0}^t(k(t-k))^n\rho(\widetilde{U})^{t-k}\rho(AP)^k\\
&\quad+\kappa_3\sum_{k=0}^t\sum_{\ell=0}^{t-k-1}(\ell(t-\ell))^n\rho(\widetilde{U})^{\ell}\rho(AP)^{t-\ell}\nonumber
\end{align}%
}}
We now distinguish two cases.
(i) Case $\rho(AP) < \rho(\widetilde U) $:
{\small{\begin{align*}
&\|s^{(i)}(t+1)-\widetilde{s}^{(i)\star}\|\leq\kappa_1t^n\rho(\widetilde{U}) ^{t+1}+\kappa_2t^n\rho(\widetilde{U})^{t}\sum_{k=0}^tk^n\bigg(\frac{\rho(AP)}{\rho(\widetilde{U})}\bigg)^k\\
&\quad+\kappa_3t^n\sum_{k=0}^t\sum_{\ell=0}^{t-k-1}\rho(\widetilde{U})^{t-k-1-\ell}\rho(AP)^{\ell+k+1}\\
&\quad\leq\kappa_1\rho(\widetilde{U}) ^{t+1}+\widetilde{\kappa}_2 t^n{\rho(\widetilde{U})^{t}}\\
&\quad+\widetilde{\kappa_3}t^n\rho(\widetilde{U})^{t}\sum_{k=0}^t(k+1)^n\frac{\rho(AP)^{k+1}}{\rho(\widetilde{U})^{k+1}}\sum_{\ell=0}^{t-k-1}\ell^n\frac{\rho(AP)^{\ell}}{\rho(\widetilde{U})^{\ell}}\\&\quad\leq\kappa_1\rho(\widetilde{U}) ^{t+1}+\kappa_4t^n\rho(\widetilde{U})^{t},
\end{align*}}}
where the last inequality, with $\kappa_4>0$, follows from Lemma~\ref{lemma:geometric_series}. The claim follows by renaming the constant terms.

(ii) Case $\rho(AP) > \rho(\widetilde U) $: Fix now $\epsilon\in(1-\rho(AP) ,1-\rho(\widetilde U) )$ then 
$\rho(AP)^s<\rho(AP)^t/(1-\epsilon)^{t-s}$ for any $s<t$. From \eqref{eq:last_expression},
there exist constants $\kappa_1>0,\kappa_2 >0$
{\small{\begin{align*}
&\|s^{(i)}(t+1)-\widetilde{s}^{(i)\star}\|\leq\kappa_1t^n\rho(\widetilde U)^{t+1}\\
&\qquad+\kappa_2 t^n\rho(AP)^{t}\sum_{k=0}^tk^n\frac{\rho(\widetilde U)^k}{(1-\epsilon)^k}\\
&\qquad+\kappa_3t^n\sum_{k=0}^t\rho(AP)^{t}\sum_{\ell=0}^{t-k-1}\ell^n\frac{\rho(\widetilde U)^{\ell}}{(1-\epsilon)^{\ell}}\nonumber\\
&\qquad\leq\kappa_1t^n\rho(\widetilde U) ^{t+1}+\kappa_4\rho(AP)^{t}t^{n+1}.
\end{align*}}}}}
Theorem~\ref{theorem:tutti_termini} follows because $\|s^{(i)}(t)-\widetilde{s}^{(i)\star}\|$ converges to~0.
\end{document}